\documentclass[lotsofwhite,charterfonts]{patmorin}
\usepackage{pdfsync}
\usepackage[sort&compress,numbers]{natbib}
\usepackage{setspace}
\usepackage{graphicx}
\usepackage{color}
\usepackage{enumitem}
\pdfoutput=1

\definecolor{grey}{gray}{0.4}

\renewcommand{\lll}{\ensuremath{\protect{(\ell,L)}}}

\newcommand{\Oh}[1]{\ensuremath{\protect\mathcal{O}(#1)}}

\usepackage{amsthm}

\newcommand{\comment}[1]{}

\newcommand{\seclabel}[1]{\label{sec:#1}}

\newcommand{\secref}[1]{\mbox{Section~\ref{sec:#1}}}

\newcommand{\figlabel}[1]{\label{fig:#1}}

\newcommand{\figref}[1]{\mbox{Figure~\ref{fig:#1}}}

\newtheorem{thm}{Theorem}{\bfseries}{\itshape}
\newcommand{\thmlabel}[1]{\label{thm:#1}}
\newcommand{\thmref}[1]{Theorem~\ref{thm:#1}}

{\bfseries}{\itshape}

\newtheorem{lem}{Lemma}{\bfseries}{\itshape}
\newcommand{\lemlabel}[1]{\label{lem:#1}}
\newcommand{\lemref}[1]{Lemma~\ref{lem:#1}}

\newtheorem{cor}{Corollary}{\bfseries}{\itshape}
\newcommand{\corlabel}[1]{\label{cor:#1}}

{\bfseries}{\itshape}

{\bfseries}{\itshape}

{\bfseries}{\itshape}

\newcommand{\etal}{\emph{et al.}}

\newcommand{\Ze}{\ensuremath{\protect{\mathbb{Z}}}}

\newcommand{\sekshun}[1]                
	{\noindent {\em 
}                            
	\markboth{#1 \hfill}{#1 \hfill} 
	}

\title{Graph Layouts via Layered Separators}
\author{Vida Dujmovi\'c
  \footnote{School of Mathematics and Statistics \& 
Department of Systems and Computer Engineering,  
   Carleton University, Ottawa, Canada
    \texttt{\{vdujmovic@math.carleton.ca\}}. 
    Research supported in part by NSERC.}
}
\date{}

\begin{document}
\maketitle

\begin{abstract}
A  $k$-queue layout of a graph consists of a total order of the
vertices, and a partition of the edges into $k$ sets such that
no two edges that are in the same set are nested with respect to the vertex ordering. A $k$-track layout of a graph consists of a vertex $k$-colouring, and a total order of each
vertex colour class, such that between each pair of colour classes no
two edges cross.  The queue-number  (track-number) of a graph $G$, is
the minimum $k$ such that $G$ has a $k$-queue ($k$-track) layout.

This paper proves that every $n$-vertex planar graph has
track number and queue number at most 
 \Oh{\log n}. 
 This improves the
result of Di Battista, Frati and Pach [{\em Foundations of Computer Science},
(FOCS '10), pp. 365--374]  who proved the first sub-polynomial bounds on
the queue number and track number of planar
graphs. Specifically, they obtained \Oh{\log^2 n} queue
number and \Oh{\log^{8} n} track number bounds for planar
graphs. 

The result also implies that every planar graph has a 3D
crossing-free grid drawing in \Oh{n\log n} volume. The proof uses a non-standard 
 type of graph separators.
\end{abstract}

\section{Introduction}

A \emph{queue layout} of a graph consists of a total order of the
vertices, and a partition of the edges into sets (called {\em queues}) such that
no two edges that are in the same set are nested with respect to the vertex ordering. The minimum number of
queues in a queue layout of a graph is its \emph{queue-number}.  Queue layouts
have been introduced by Heath, Leighton, and Rosenberg \cite{HR-SJC92, HLR-SJDM92}
and have been extensively studied since \cite{EI71, Hasunuma-GD03, HLR-SJDM92,
  HR-SJC92, Pemmaraju-PhD, RM-COCOON95, SS00, Tarjan72a, DMW-SJC05, 
  DujWoo-DMTCS05, DPW-DMTCS04, DBLP:journals/dmtcs/Wood08,
  dfp-qnpg-focs-10, DBLP:journals/dmtcs/Wood05a}. They have  applications in parallel process scheduling, fault-tolerant
processing, matrix computations,  and sorting networks (see
\cite{Pemmaraju-PhD} for a survey). Queue layouts of directed acyclic
graphs  \cite{BCLR-JPDC96, HP-SJC99, HPT-SJC99, Pemmaraju-PhD} and
posets \cite{HP-SJDM97, Pemmaraju-PhD} have also been investigated. 

The dual concept of a queue layout is a \emph{stack layout}, introduced by
Ollmann~\cite{Ollmann73} and commonly called a \emph{book
  embedding}. It is defined similarly, except that no two edges in
the same set are allowed to cross with respect to the vertex
ordering. {\em Stack number} (also known as {\em book thickness}) is
known to be bounded for planar graphs \cite{Yannakakis89}, bounded
genus graphs \cite{Malitz94b} and, most generally, all proper minor
closed graph families \cite{Blankenship-PhD03, BO01}. 

No such bounds are known for the queue number of these graph families. The question of Heath~\etal~\cite{HR-SJC92,HLR-SJDM92} on whether every
planar graph has \Oh{1} queue-number, and the  more general question
(since planar graphs have stack-number at most four
\cite{Yannakakis89}) of whether queue-number is bounded by
stack-number remains open. Heath~\etal~\cite{HR-SJC92,HLR-SJDM92}
conjectured that both of these questions have an affirmative answer. Until recently, the best known bound for the queue number of planar graphs was \Oh{\sqrt{n}}. This upper bound follows easily from the fact that planar graphs have pathwidth at most \Oh{\sqrt{n}}. 
 In a recent breakthrough \cite{dfp-qnpg-focs-10}, this queue number
 bound for planar graphs was reduced to \Oh{\log^2 n}, by Di Battista,
 Frati and Pach\footnote{The original bound proved in this conference
   paper, \cite{dfp-qnpg-focs-10}, is \Oh{\log^4n}. The bounds stated
   here are from the journal version that is under the submission.} \cite{dfp-qnpg-focs-10}. The proof,  however, is quite
 involved and long. 

We improve the bound for the queue number of planar
graphs  to \Oh{\log n}.
%
%
Pemmaraju~\cite{Pemmaraju-PhD} 
conjectured that planar graphs have \Oh{\log n} queue-number. Thus the result answers this question in affirmative. He also conjecture that
this is the correct lower bound. To date, however, the best known lower bound
is a constant. 

The proof is simple and it uses a special kind of graph separators. In
particular, the main result states that every $n$-vertex graph that has such
a separator (and it turns out that planar graphs do) has an \Oh{\log
  n} queue number 
As such, the result may provide a
tool for breaking the \Oh{\sqrt{n}} queue number bound for other graph families,
such as graphs of bounded genus and other proper minor closed
families of graphs.

One of the motivations for studying queue layouts is their connection
with three-dimensional graph drawings in a grid of small volume. In particular, a
{\em 3D grid drawing} of a graph is a placement of the vertices at
distinct points in $\Ze^3$, such that the line-segments representing
the edges are pairwise non-crossing. A 3D grid drawing that fits in an
axis-aligned box with side lengths $X-1$, $Y-1$, and $Z-1$, is a  $X
\times Y \times Z$ drawing with volume $X \cdot Y\cdot Z$.

It has been established in \cite{DMW-SJC05, DPW-DMTCS04}, that an $n$-vertex graph $G$ has an $\Oh{1}\times \Oh{1}\times \Oh{n}$ drawing, if and only if $G$ has
\Oh{1} queue-number. Therefore, if a graph has a bounded queue number
then it has a linear volume 3D grid drawing. One of the most extensively
studied graph drawing questions is whether planar graphs have linear
volume 3D grid drawings -- the question is due to Felsner
\etal\cite{FLW-GD01-ref}. Our results imply \Oh{n\log n} bound, improving on the previous \Oh{n\log^{8} n} bound
\cite{dfp-qnpg-focs-10}.

In the next section, we give precise statement of our result and
introduce a tool used to obtain it. In \secref{proof} we prove the main
result and then conclude with some open problems in \secref{conclusion}.

\section{Results and Tools}\seclabel{results}


The main tool in proving our result is the following type of graph
separators. 

A \emph{layering} of a graph $G$ is a partition $V_0,V_1,\dots,V_p$ of
$V(G)$ such that for every edge $vw\in E(G)$, if $v\in V_i$ and $w\in
V_j$ then $|i-j|\leq 1$. Each set $V_i$ is called a \emph{layer}.  A
\emph{separation} of a graph $G$ is a pair $(G_1,G_2)$ of subgraphs of
$G$, such that $G=G_1\cup G_2$ and  there is no edge of
$G$ between $V(G_1)-V(G_2)$ and $V(G_2)-V(G_1)$. 

A graph $G$ has a \emph{layered} $\ell$-\emph{separator} if for some fixed
layering $L$ of  its vertices the following holds: For every subgraph $G'\subseteq G$ there is a separation $(G'_1,G'_2)$ of $G'$
such that each layer of $L$ contains at most $\ell$ vertices in
$V(G'_1)\cap    V(G'_2)$, and both $V(G'_1)-V(G'_2)$ and
$V(G'_2)-V(G'_1)$ contain at most $\frac{2}{3}|V(G')|$ vertices. Here the set  $V(G'_1)\cap V(G_2')$ is a 
(layered $\ell$--) \emph{separator} of $G'$.  Finally, if a graph $G$
has a \emph{layered} $\ell$-\emph{separator} for some fixed layering $L$, we say that $G$ has \lll--\emph{separator}. 
Note that these separators do not necessarily have small order, in
particular $V(G'_1)\cap V(G'_2)$ can have linear number of
vertices of $G'$.


The notion of layered separators is not new. They were used implicitly,
for example, in the famous proof, by Lipton and Tarjan \cite{LT-SJAM79}, that planar graphs have a separator of order
\Oh{\sqrt{n}}. Specifically, consider a breath-first-search tree $T$ of a
graph $G$ and the layering $L$ defined by partitioning the vertices of
$G$ according to their distance to the root of $T$. Each edge that is not in $T$ defines a
unique cycle, called a {\em $T$-cycle}. One step in their proof was
to show that any edge maximal planar graph has a $T$-cycle separator. A $T$-cycle contains at most
two vertices from each layer of $L$. However, to jump from the
existence of a $T$-cycle separator of $G$ to the existence of a layered
$2$-separator of $G$ requires more work. In particular, consider a connected
component $G'$ of $G$ that remains after removing a
$T$-cycle separator from $G$. In order to apply the result of Lipton
and Tarjan to $G'$, edges may need to be added to $G'$  in
such a way that it remains planar and such that $L$ is still its breath first
search layering. This is (at
least in the case of planar graphs) possible and the explicit
proof can be found in \cite{DFJW-2012}, where the layered separators in
this form have been introduced. The authors used layered separators to show that planar
graphs have non-repetitive chromatic number at most \Oh{\log n} \cite{DFJW-2012}, thus breaking a long
standing \Oh{\sqrt{n}} bound. 




\begin{lem}\lemlabel{sep}\cite{LT-SJAM79, DFJW-2012}
Let $L$ be a breath first search layering of a triangulated (that is,
edge maximal) planar graph $G$. Then $G$ has a layered
$(2,L)$-separator.
\end{lem}









%

Our main result is expressed in terms of track layouts of graphs, a type of graph layouts that is closely related to queue layouts and 3D grid drawings. We define track layouts first. 
A \emph{vertex} $|I|$\emph{-colouring} of a graph $G$ is a partition
$\{V_i:i\in I\}$ of $V(G)$ such that for every edge $vw\in E(G)$, if $v\in V_i$
and $w\in V_j$ then $i\ne j$. The elements of the set $I$ are \emph{colours}, and  each
set $V_i$ is a \emph{colour class}. Suppose that  $<_i$ is a total order on
each colour class $V_i$.  Then each pair $(V_i,<_i)$ is a \emph{track},  and
$\{(V_i,<_i):i\in I\}$ is an $|I|$-\emph{track assignment} of $G$. The \emph{span} of an edge $vw$ in a track assignment $\{(V_i, <_i) :1\leq i \leq t\}$  is $|i-j|$ where $v\in V_i$ and $w\in V_j$. 

An \emph{X-crossing} in a track assignment consists of two edges $vw$ and $xy$
such that $v<_ix$ and $y<_jw$, for distinct colours $i$ and $j$.  A
$t$-track assignment of $G$ that has no X-crossings is called
$t$-\emph{track layout} of $G$. The minimum $t$ such that a graph $G$
has $t$-track layout is called \emph{track number} of $G$.


The main result of this paper is the following. 

\begin{thm}\thmlabel{main}
Every $n$-vertex graph $G$ that has a layered $\ell$-separator has track number at most 
$3\ell\,\lceil\log_{3/2}n\rceil+3\ell$.
\end{thm}

All the other results mentioned earlier follow from this theorem and previously known
results. Before proving \thmref{main} we discuss these implications.

If a graph G has a $t$-track layout with maximum edge span $s$, then
the queue number of $G$ is at most $s$ and thus at most $t-1$
\cite{DMW-SJC05}. Furthermore, every $c$-colourable $t$-track
graph $G$ with $n$ vertices has a 3D grid drawing in \Oh{t^2 n} volume \cite{DMW-SJC05}
as well as in \Oh{c^7 t n} volume \cite{DujWoo-SubQuad-AMS}. Thus \thmref{main} implies.

\begin{cor}\corlabel{cor1}
Every $n$-vertex $c$-colourable graph $G$ that has a layered
$\ell$-separator has queue number at most $3\ell\,\lceil\log_{3/2}n\rceil+3\ell$ and a 3D grid drawing in {O($\log n$)}  volume as well as in \Oh{c^7\ell\, n\log n} volume. 
\end{cor}

Together with \lemref{sep} this finally implies all the claimed
results on planar graphs.

\begin{cor}\corlabel{planar}
Every $n$-vertex planar graph has track number and queue number at
most $6\lceil\log_{3/2}n\rceil+6$ and a 3D grid drawing in \Oh{n\log n} volume.
\end{cor}

 

\section{Proof of \thmref{main}}\seclabel{proof}

Let $G$ be an $n$-vertex graph, let $L=\{V_0,V_1,\dots,V_p\}$ be a
  layering of $G$, and let $\ell\geq 1$, such that $G$ has
  \lll--separator.  Removing such a separator from $G$ splits $G$ into
  connected components each of which has at most $\frac{2}{3} |V(G)|$
  vertices and its own \lll--separator. Thus the process can continue until
each component is an \lll--separator of itself. This process naturally
defines a rooted tree $S$ and a mapping of $V(G)$ to the nodes of $S$, as
follows. The root of $S$ is a node  to which the vertices of an
\lll--separator of $G$ are mapped. The root has $c\geq 1$ children in
$S$, one for each connected component $G_j$, $j\in [1,c]$, obtained by
removing the \lll--separator from $G$. The vertices of an \lll--separator of $G_j$, $j\in [1,c]$, are mapped to a child of the
root. The process continues until each component is an \lll--separator
of itself, or more specifically until each component has at most
$\ell$-vertices in each layer of $L$. In that case, such a component
is an  \lll--separator of itself and its vertices are mapped to a leaf of $S$. This
defines a rooted tree $S$ and a partition of $V$ to the nodes of $S$. One important
observation, is that the height of $S$ is most $\lceil\log_{3/2} n\rceil+1$. 
For a node $s$ of $S$, let $s(G)$ denote the set of vertices of $G$
that are mapped to $s$ and let $G[s]$ denote the graph induced by
$s(G)$ in $G$. Note that for each node $s$ of $S$, $s(G)$ has at most $\ell$ vertices in any layer of $L$.



\thmref{main} states that $G$ has a track layout with \Oh{\ell \log n}
tracks. To prove this we will first create a track layout $T$ of $G$ with
possibly lots of tracks. We then modify that layout in order to reduce the
number of tracks to \Oh{\ell \log n}.

To ease the notation, for a track $(V_r,<_r)$, indexed by colour $r$,
in a track assignment $R$, we denote that track by $(r)$  when the
ordering on each colour class is implicit.  Also we sometimes write
$v<_R w$. This indicates that $v$ and $w$ are on a same track $r$ of
$R$ and that $v<_r w$.

Throughout this section, it is important to keep in mind that a
\emph{layer} is a subset of vertices of $G$ defined by the layering
$L$ and that a \emph{track} is an (ordered) subset of vertices of $G$
defined by a track assignment of $G$.

We first define a track assignment $T$ of $G$. Consult
\figref{track-ass} in the process. Each vertex $v$ of  $V(G)$
is assigned to a track whose colour is defined by three indices $(d,
i, k)$.  Let $s_v$ denote the node of the tree $S$ that $v$ is mapped to. The first
index is the depth of $s_v$ in $S$.
 The root is considered to have depth $1$. Thus the first index, $d$, ranges from $1$ to $\lceil \log_{3/2} n\rceil+1$. The
second index is the layer of $L$ that contains $v$. Thus the second
index, $i$, can be as big as $\Omega(n)$. Finally, $s_v(G)$ contains
at most $\ell$ vertices from layer $i$ in $L$. Label these, at most $\ell$, vertices arbitrarily from $1$ to $\ell$ and let the third
index $k$ of each of them be
determined by this label.  Consider the tracks  themselves to be
lexicographically ordered. 

To complete the track assignment we need to define the
ordering of vertices in the same track. To do that we first define a simple
track layout of the tree $S$. Consider a natural way to draw $S$ in
the plane without crossings such that all the nodes of $S$ that are at the same
distance from the root are drawn on the same horizontal line, as
illustrated in \figref{tree}. This
defines a track layout $T_S$ of $S$ where each horizontal line is a
track and the ordering of the nodes within each track is implied by
the crossing free drawing of $S$. 

\begin{figure}
  \begin{center} 
   \includegraphics[width=0.55\linewidth]{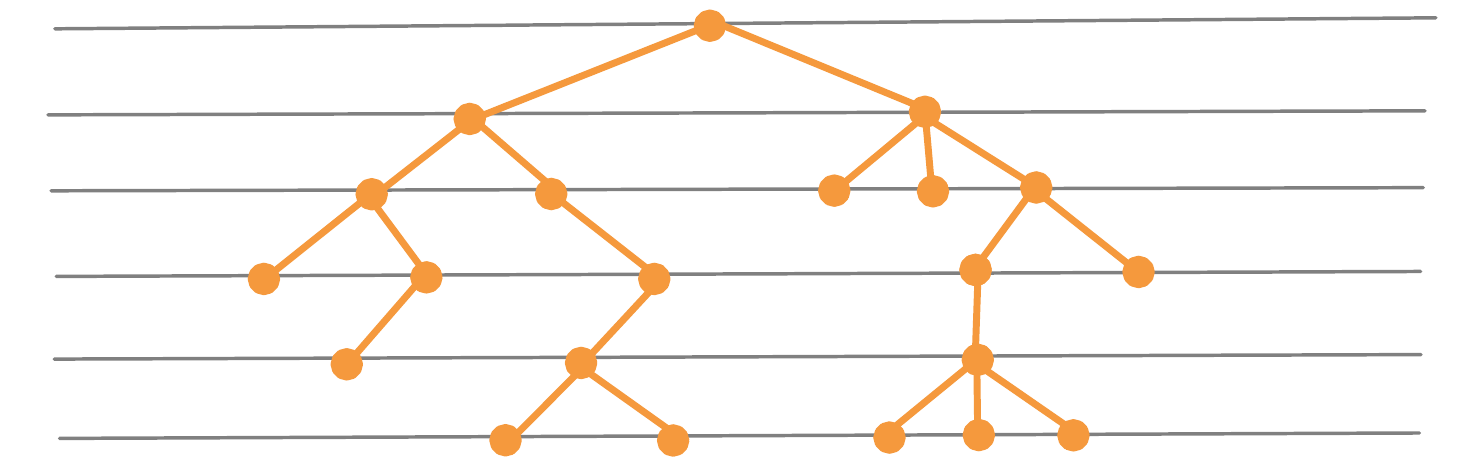}
  \end{center}
  \caption{A track layout of the tree $S$.}
  \figlabel{tree}
\end{figure}

To complete the track assignment $T$, we need to define the total
order of vertices that are in the same track of $T$. 
For any two vertices $v$ and $w$ of
$G$ that are assigned to the same track $(d, i, k)$ in $T$, let 
$v<_T w$ if the node $s_v$ that $v$ maps to in $S$ appears in $T_S$ to the left of the
node $s_w$ that $w$ maps to in $S$, that is, if $s_v<_{T_S}
s_w$. Since $v$ and $w$ are in the same track of $T$ only if they are mapped two distinct nodes of $S$ that are the the same distance from the root of $S$, this defines a total order of
each track in $T$. \figref{track-ass} depicts the resulting track assignment $T$ of $G$.

\begin{figure}
~\vspace{-2cm}
  \begin{center} 
   \includegraphics[width=0.9\linewidth]{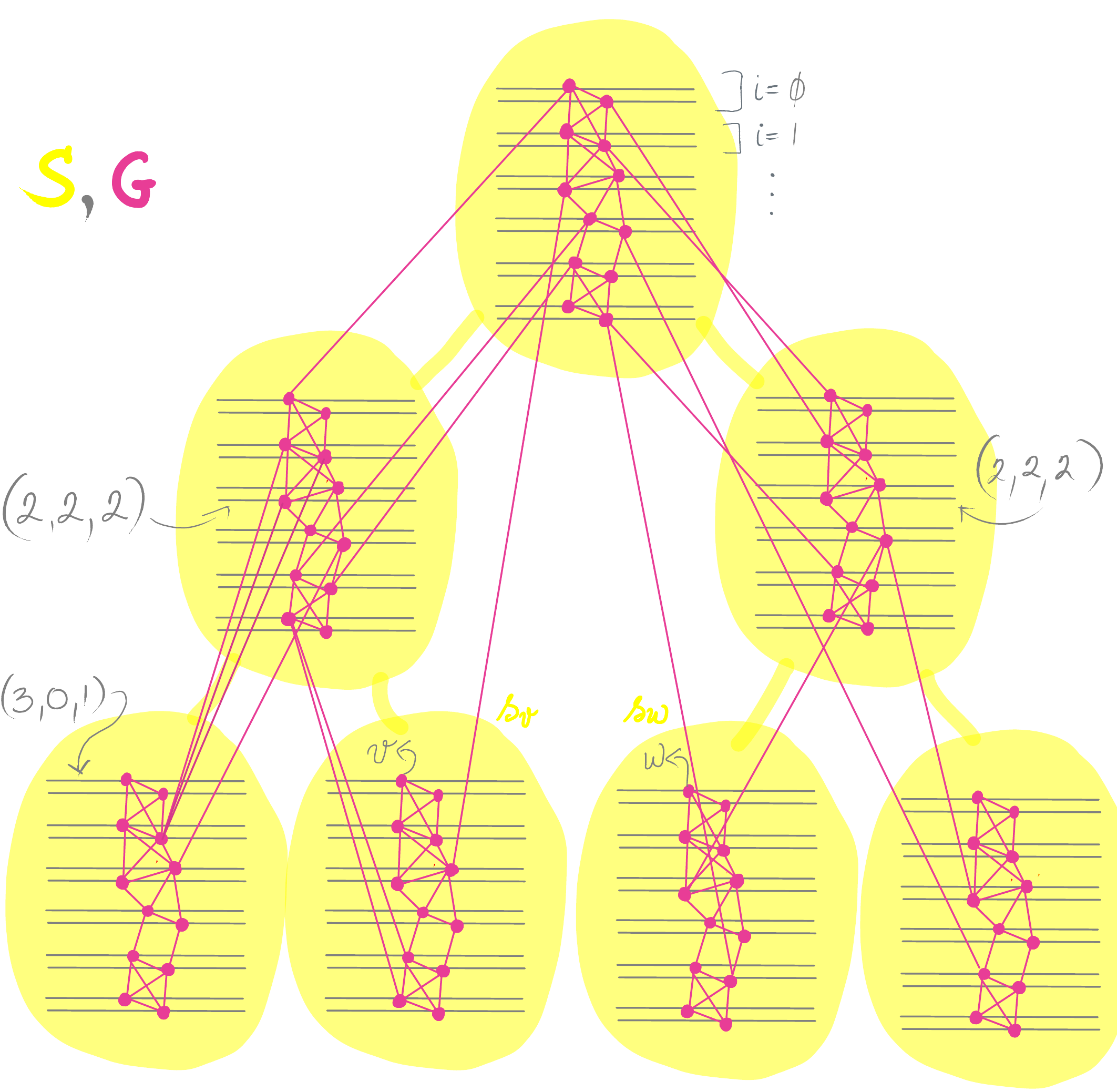}
  \end{center}
  \caption{A track layout $T$ of a graph $G$ which has a layered $(\ell=2$)--separator.}
  \figlabel{track-ass}
 \begin{center} 
   \includegraphics[width=0.8\linewidth]{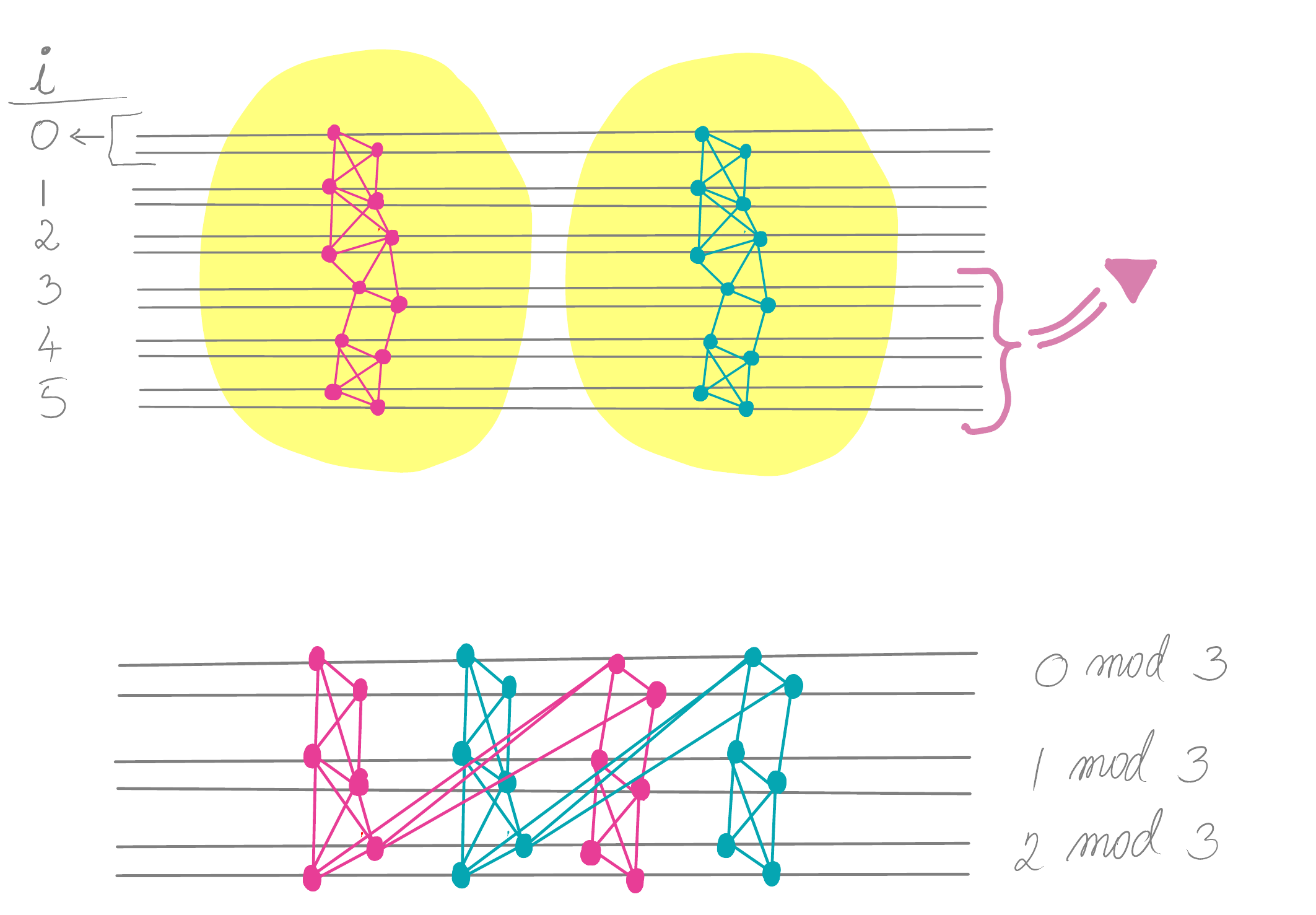}
  \end{center}
  \caption{Top figure: the track layout $T_2$ of $G_2$.
    Bottom figure: the track layout $T'_2$ obtained by wrapping $T_2$.}
 \figlabel{d-wrap}
\end{figure}

It is not difficult to verify that $T$ is indeed a track layout of
$G$, that is, $T$ does not have X-crossings. This track layout however
may have $\Omega(n)$ tracks. We now modify $T$ to reduce the number of
tracks to the claimed number. 

For a vertex $v$ of $G$, let  $(d_v, i_v, k_v)$ denote the track of $v$
in $T$.


Dujmovi\'c and Wood \cite{DPW-DMTCS04}, inspired by Felsner~\etal 
\cite{FLW-JGAA03}, proved a simple wrapping lemma that says that a track layout
with maximum span $s$ can be wrapped into a $(2s+1)$--track layout. 
For example, \figref{tree-wrap} in the appendix, depicts how a track layout of a tree with maximum
span $1$ (like the track layout $T_S$ of $S$) can be wrapped around a 
triangular prism to give a 3-track layout.

Unfortunately, the track layout $T$ of $G$ does not have a bounded span
-- its span can be $\Omega(n)$. (Since the tracks of $T$ are ordered by
lexicographical ordering, span is well defined in $T$.)  
However parts of the layout do have bounded span. In particular, consider the graph,
$G_d$ induced by the vertices of $G$ that are assigned to the tracks
of $T$ that have the same first index, $d$. 
For each $d$, the tracks of $T$ with that first
index equal to $d$, define a track layout $T_d$ of $G_d$, as
illustrated, for $d=2$ case, in the top part of
\figref{d-wrap}. Recall that $G_d$ is comprised of disjoint layered
$(l, L)$-separators
(see \figref{track-ass}). Since each edge in a layered $(l, L)$-separator either
connects two  vertices in the same layer of $L$  or two vertices from
two consecutive layers of $L$,  the span of an edge of $G_d$ in $T_d$ is at most 
 $2\ell-1$.

For each $d$, 
we now wrap the track layout $T_d$ into a  $3\ell$-track layout $T'_d$ of
$G_d$, as illustrated in the bottom of \figref{d-wrap}.  
The exact version of the wrapping lemma we use is given in the appendix, see
\lemref{wrap}. It mimics the wrapping of  Felsner~\etal\cite{FLW-JGAA03} and it is 
included for completeness. 
This defines a track assignment $T'$ of $G$. 

(One may be tempted to, instead of wrapping all of $G_d$, wrap
individually the
track layouts of the graphs induced by the vertices mapped to the same
node of $S$. It can be shown however that such strategy can introduce
X-crossings in $T'$. )

The wrapping lemma, \lemref{wrap}, implies that for all $d$, $T'_d$ has the following useful properties. Consider two
vertices $a$ and $b$ that are in the same track $f'=(d_a, i_a \mod 3,
k_a)=(d_b, i_b \mod 3, k_b)$ in $T'_d$.
Then if $a<_{f'} b$ in $T'_d$ and  
\begin{enumerate}
\item[(1)] $i_a\ne i_b$ then $i_b\geq i_a+3$, 
\item[(2)] otherwise, ($i_a=i_b$), $a$ and $b$ were in the same track $f$ in $T$
  and $a<_f b$. (This is because the wrapping does not change the
  ordering of vertices that were already in the same track in $T$).  
\end{enumerate} 

Since $d\leq \lceil \log_{3/2} n\rceil+1$, $i\mod 3\leq 3$ and $k\leq \ell$, 
the  track assignment $T'$ of $G$ has at most  
$3\ell\, \lceil\log_{3/2} n\rceil+1$ tracks, as claimed. It remains to prove that $T'$ is in fact a
track layout of $G$, that is, there are no X-crossings in the track
assignment $T'$

Assume by contradiction that there are two edges $vw$ and
$xy$ that form an X-crossing in $T'$. Let  $v$ and $x$ belong to a
same track in $T'$ and let $y$ and $w$ belong to a same track in $T'$. If
$d_v=d_w=d_x=d_y$, then $v, w, x$ and $y$ belong to the the same graph
$G_d$ and thus they do not form and X-crossing since $T'_d$ does not
have X-crossings by the wrapping lemma.

Thus $d_v=d_x=d_1$ and $d_w=d_y=d_2$ and $d_1\ne d_2$. Let without
loss of generality $d_1<d_2$ and $v<_{T'}x$ and
$y<_{T'} w$ in $T'$. Since $w$ and $y$ are in the same track,  $d_y=d_w$, $k_y=k_w$ and either $i_y=i_w$
or $i_w\geq i_y+3$ by properties (1) and (2). There are thus two cases to consider. First
consider the case that $i_w\geq i_y+3$. Since $w$ is adjacent to $v$,
$i_v=\{i_w-1, i_w, i_w+1\}$ and similarly $i_x=\{i_y-1, i_y,
i_y+1\}$. Thus $i_v\geq  i_w-1\geq i_y+2$ and $i_x\leq i_y+1$. Thus,
$i_v> i_x$ and property $(1)$ applies to $v$ and $x$. This contradicts
the assumption that $v<_{T'}x$, since property $(1)$, implies that
$i_x>i_v$.


Finally, consider the case that $i_y=i_w$. Then $y$ and $w$ are in the
same track in $T$ and their ordering, $y<_{T'}w$, in $T'$ is the same
as in $T$, $y<_{T}w$, by
property $(2)$. Since $v$ and $x$ are in the same track in $T'$ and
$v<_{T'} x$, either $i_v=i_x$ or
$i_x\geq i_v+3$, by properties (1) and (2). Thus again, since $w$ is
adjacent to $v$, $i_v=\{i_w-1, i_w, i_w+1\}$ and similarly
$i_x=\{i_y-1, i_y,i_y+1\}$. Since $i_y=i_w$, no pair of these indices differs by at
least 3 and thus $i_v=i_x$ by property $(1)$. That implies that $v$ and $x$ are in the same track in $T$
and thus by  property $(2)$ their
ordering in $T'$, $v<_{T'} x$, is the same as in $T$, $v<_T x$. This implies that
$vw$ and $xy$ form an X-crossing in $T$ thus providing the desired 
contradiction. This completes the proof of \thmref{main}. $\qed$






\section{Conclusions}\seclabel{conclusion}

Classical separators are one of the most powerful and widely used
tools in graph theory as they allow for solving all kinds of
combinatorial problems with the divide and conquer method. 
However, for most of the interesting graph classes such separators can
be polynomial in $n$, such as $\Omega(\sqrt{n})$ for planar graphs,
thus it is unclear how they can help in proving sub-polynomial
bounds. 

Layered separators provide extra structure that can aid in attacking
problems that inherently look for some vertex ordering, such as in
queue layouts and to lesser extent in non-repetitive colourings -- and
this despite of the fact that such separators can be of linear
order. The critical feature of layered separators is that edges only
appear between consecutive layers and that the number of vertices per
layer is bounded. As we have seen, planar graphs have such
separators. It is an interesting question to determine what other
classes of graphs have layered \Oh{1}-separators, especially given the
implications such separators have on the aforementioned problems. 
The graphs of bounded genus and more generally all  proper minor closed families are the natural
candidates.

Related to that, is the result of Gilbert, Hutchinson and Tarjan (see the proof of Theorem $4$
\cite{GHT-JAlg84}) who showed that every graph $G$ of genus $g$ has $2g$
$T$-cycles whose removal from $G$ leaves a planar graph $G'$. However,
(recall the discussion from \secref{results}) it is unclear if $G'$
can be made edge-maximal planar while keeping $L$ as its breath first
search layering.

Finally, the way the layered \Oh{1}-separators are used in this paper, 
and in \cite{DFJW-2012}, inherently leads to logarithmic upper
bounds. It seems difficult to envision how such separators could be
used to obtain sub-logarithmic upper bounds. Thus, if these \Oh{\log n} upper bounds are not
tight, the new methods may be needed to break them.






\bibliographystyle{myNatbibStyle}
\bibliography{myBibliography,myConferences}

\def\soft#1{\leavevmode\setbox0=\hbox{h}\dimen7=\ht0\advance \dimen7
  by-1ex\relax\if t#1\relax\rlap{\raise.6\dimen7
  \hbox{\kern.3ex\char'47}}#1\relax\else\if T#1\relax
  \rlap{\raise.5\dimen7\hbox{\kern1.3ex\char'47}}#1\relax \else\if
  d#1\relax\rlap{\raise.5\dimen7\hbox{\kern.9ex \char'47}}#1\relax\else\if
  D#1\relax\rlap{\raise.5\dimen7 \hbox{\kern1.4ex\char'47}}#1\relax\else\if
  l#1\relax \rlap{\raise.5\dimen7\hbox{\kern.4ex\char'47}}#1\relax \else\if
  L#1\relax\rlap{\raise.5\dimen7\hbox{\kern.7ex
  \char'47}}#1\relax\else\message{accent \string\soft \space #1 not
  defined!}#1\relax\fi\fi\fi\fi\fi\fi} \def\Dbar{\leavevmode\lower.6ex\hbox to
  0pt{\hskip-.23ex \accent"16\hss}D}
\begin{thebibliography}{29}
\providecommand{\natexlab}[1]{#1}
\providecommand{\url}[1]{\texttt{#1}}
\providecommand{\urlprefix}{}
\expandafter\ifx\csname urlstyle\endcsname\relax
  \providecommand{\doi}[1]{doi:\discretionary{}{}{}#1}\else
  \providecommand{\doi}{doi:\discretionary{}{}{}\begingroup
  \urlstyle{rm}\Url}\fi

\bibitem[{Battista et~al.(2010)Battista, Frati, and Pach}]{dfp-qnpg-focs-10}
\textsc{Giuseppe~Di Battista, Fabrizio Frati, and Janos Pach}.
\newblock On the queue number of planar graphs.
\newblock In \emph{Foundations of Computer Science (FOCS '10)}, pp. 365--374.
  2010.

\bibitem[{Bhatt et~al.(1996)Bhatt, Chung, Leighton, and
  Rosenberg}]{BCLR-JPDC96}
\textsc{Sandeep~N. Bhatt, Fan R.~K. Chung, F.~Thomson Leighton, and Arnold~L.
  Rosenberg}.
\newblock Scheduling tree-dags using {FIFO} queues: A control-memory trade-off.
\newblock \emph{J. Parallel Distrib. Comput.}, 33:55--68, 1996.

\bibitem[{Blankenship(2003)}]{Blankenship-PhD03}
\textsc{Robin Blankenship}.
\newblock \emph{Book Embeddings of Graphs}.
\newblock Ph.D. thesis, Department of Mathematics, Louisiana State University,
  U.S.A., 2003.

\bibitem[{Blankenship and Oporowski(2001)}]{BO01}
\textsc{Robin Blankenship and Bogdan Oporowski}.
\newblock Book embeddings of graphs and minor-closed classes.
\newblock In \emph{Proc. 32nd Southeastern International Conf. on
  Combinatorics, Graph Theory and Computing}. Department of Mathematics,
  Louisiana State University, 2001.

\bibitem[{Dujmovi{\'c} et~al.(2012)Dujmovi{\'c}, Frati, Joret, and
  Wood}]{DFJW-2012}
\textsc{Vida Dujmovi{\'c}, Fabrizio Frati, Gwena{\H{e}}l Joret, and David~R.
  Wood}.
\newblock Nonrepetitive colourings of planar graphs with ${O}(\log n)$ colours,
  2012.
\newblock \urlprefix\url{http://arxiv.org/abs/1202.1569}.

\bibitem[{Dujmovi{\'c} et~al.(2005)Dujmovi{\'c}, Morin, and Wood}]{DMW-SJC05}
\textsc{Vida Dujmovi{\'c}, Pat Morin, and David~R. Wood}.
\newblock Layout of graphs with bounded tree-width.
\newblock \emph{SIAM J. Comput.}, 34(3):553--579, 2005.

\bibitem[{Dujmovi{\'c} et~al.(2004)Dujmovi{\'c}, P\'or, and Wood}]{DPW-DMTCS04}
\textsc{Vida Dujmovi{\'c}, Attila P\'or, and David~R. Wood}.
\newblock Track layouts of graphs.
\newblock \emph{Discrete Math. Theor. Comput. Sci.}, 6(2):497--522, 2004.

\bibitem[{Dujmovi{\'c} and Wood(2004)}]{DujWoo-SubQuad-AMS}
\textsc{Vida Dujmovi{\'c} and David~R. Wood}.
\newblock Three-dimensional grid drawings with sub-quadratic volume.
\newblock In \textsc{J\'{a}nos Pach}, ed., \emph{Towards a Theory of Geometric
  Graphs}, vol. 342 of \emph{Contemporary Mathematics}, pp. 55--66. Amer. Math.
  Soc., 2004.

\bibitem[{Dujmovi{\'c} and Wood(2005)}]{DujWoo-DMTCS05}
\textsc{Vida Dujmovi{\'c} and David~R. Wood}.
\newblock Stacks, queues and tracks: Layouts of graph subdivisions.
\newblock \emph{Discrete Math. Theor. Comput. Sci.}, 7:155--202, 2005.

\bibitem[{Even and Itai(1971)}]{EI71}
\textsc{Shimon Even and A.~Itai}.
\newblock Queues, stacks, and graphs.
\newblock In \textsc{Zvi Kohavi and Azaria Paz}, eds., \emph{Proc.
  International Symp. on Theory of Machines and Computations}, pp. 71--86.
  Academic Press, 1971.

\bibitem[{Felsner et~al.(2002)Felsner, Liotta, and Wismath}]{FLW-GD01-ref}
\textsc{Stefan Felsner, Giussepe Liotta, and Stephen~K. Wismath}.
\newblock Straight-line drawings on restricted integer grids in two and three
  dimensions.
\newblock In \emph{Proc. 9th International Symp. on Graph Drawing (GD '01)},
  vol. 2265 of \emph{Lecture Notes in Comput. Sci.}, pp. 328--342. Springer,
  2002.

\bibitem[{Felsner et~al.(2003)Felsner, Liotta, and Wismath}]{FLW-JGAA03}
\textsc{Stefan Felsner, Giussepe Liotta, and Stephen~K. Wismath}.
\newblock Straight-line drawings on restricted integer grids in two and three
  dimensions.
\newblock \emph{J. Graph Algorithms Appl.}, 7(4):363--398, 2003.

\bibitem[{Gilbert et~al.(1984)Gilbert, Hutchinson, and Tarjan}]{GHT-JAlg84}
\textsc{John~R. Gilbert, Joan~P. Hutchinson, and Robert~E. Tarjan}.
\newblock A separator theorem for graphs of bounded genus.
\newblock \emph{J. Algorithms}, 5(3):391--407, 1984.

\bibitem[{Hasunuma(2004)}]{Hasunuma-GD03}
\textsc{Toru Hasunuma}.
\newblock Laying out iterated line digraphs using queues.
\newblock In \textsc{Giuseppe Liotta}, ed., \emph{Proc. 11th International
  Symp. on Graph Drawing (GD '03)}, vol. 2912 of \emph{Lecture Notes in Comput.
  Sci.}, pp. 202--213. Springer, 2004.

\bibitem[{Heath et~al.(1992)Heath, Leighton, and Rosenberg}]{HLR-SJDM92}
\textsc{Lenwood~S. Heath, F.~Thomson Leighton, and Arnold~L. Rosenberg}.
\newblock Comparing queues and stacks as mechanisms for laying out graphs.
\newblock \emph{SIAM J. Discrete Math.}, 5(3):398--412, 1992.

\bibitem[{Heath and Pemmaraju(1997)}]{HP-SJDM97}
\textsc{Lenwood~S. Heath and Sriram~V. Pemmaraju}.
\newblock Stack and queue layouts of posets.
\newblock \emph{SIAM J. Discrete Math.}, 10(4):599--625, 1997.

\bibitem[{Heath and Pemmaraju(1999)}]{HP-SJC99}
\textsc{Lenwood~S. Heath and Sriram~V. Pemmaraju}.
\newblock Stack and queue layouts of directed acyclic graphs. {II}.
\newblock \emph{SIAM J. Comput.}, 28(5):1588--1626, 1999.

\bibitem[{Heath et~al.(1999)Heath, Pemmaraju, and Trenk}]{HPT-SJC99}
\textsc{Lenwood~S. Heath, Sriram~V. Pemmaraju, and Ann~N. Trenk}.
\newblock Stack and queue layouts of directed acyclic graphs. {I}.
\newblock \emph{SIAM J. Comput.}, 28(4):1510--1539, 1999.

\bibitem[{Heath and Rosenberg(1992)}]{HR-SJC92}
\textsc{Lenwood~S. Heath and Arnold~L. Rosenberg}.
\newblock Laying out graphs using queues.
\newblock \emph{SIAM J. Comput.}, 21(5):927--958, 1992.

\bibitem[{Lipton and Tarjan(1979)}]{LT-SJAM79}
\textsc{Richard~J. Lipton and Robert~E. Tarjan}.
\newblock A separator theorem for planar graphs.
\newblock \emph{SIAM J. Appl. Math.}, 36(2):177--189, 1979.

\bibitem[{Malitz(1994)}]{Malitz94b}
\textsc{Seth~M. Malitz}.
\newblock Genus $g$ graphs have pagenumber ${O}(\sqrt g)$.
\newblock \emph{J. Algorithms}, 17(1):85--109, 1994.

\bibitem[{Ollmann(1973)}]{Ollmann73}
\textsc{L.~Taylor Ollmann}.
\newblock On the book thicknesses of various graphs.
\newblock In \textsc{Frederick Hoffman, Roy~B. Levow, and Robert S.~D. Thomas},
  eds., \emph{Proc. 4th {S}outheastern {C}onference on {C}ombinatorics, {G}raph
  {T}heory and {C}omputing}, vol. VIII of \emph{Congr. Numer.}, p. 459.
  Utilitas Math., 1973.

\bibitem[{Pemmaraju(1992)}]{Pemmaraju-PhD}
\textsc{Sriram~V. Pemmaraju}.
\newblock \emph{Exploring the Powers of Stacks and Queues via Graph Layouts}.
\newblock Ph.D. thesis, Virginia Polytechnic Institute and State University,
  U.S.A., 1992.

\bibitem[{Rengarajan and Veni~Madhavan(1995)}]{RM-COCOON95}
\textsc{S.~Rengarajan and C.~E. Veni~Madhavan}.
\newblock Stack and queue number of $2$-trees.
\newblock In \textsc{Ding-Zhu Du and Ming Li}, eds., \emph{Proc. 1st Annual
  International Conf. on Computing and Combinatorics (COCOON '95)}, vol. 959 of
  \emph{Lecture Notes in Comput. Sci.}, pp. 203--212. Springer, 1995.

\bibitem[{Shahrokhi and Shi(2000)}]{SS00}
\textsc{Farhad Shahrokhi and Weiping Shi}.
\newblock On crossing sets, disjoint sets, and pagenumber.
\newblock \emph{J. Algorithms}, 34(1):40--53, 2000.

\bibitem[{Tarjan(1972)}]{Tarjan72a}
\textsc{Robert~E. Tarjan}.
\newblock Sorting using networks of queues and stacks.
\newblock \emph{J. Assoc. Comput. Mach.}, 19:341--346, 1972.

\bibitem[{Wood(2005)}]{DBLP:journals/dmtcs/Wood05a}
\textsc{David~R. Wood}.
\newblock Queue layouts of graph products and powers.
\newblock \emph{Discrete Mathematics {\&} Theoretical Computer Science},
  7(1):255--268, 2005.

\bibitem[{Wood(2008)}]{DBLP:journals/dmtcs/Wood08}
\textsc{David~R. Wood}.
\newblock Bounded-degree graphs have arbitrarily large queue-number.
\newblock \emph{Discrete Mathematics {\&} Theoretical Computer Science}, 10(1),
  2008.

\bibitem[{Yannakakis(1989)}]{Yannakakis89}
\textsc{Mihalis Yannakakis}.
\newblock Embedding planar graphs in four pages.
\newblock \emph{J. Comput. System Sci.}, 38(1):36--67, 1989.

\end{thebibliography}

\newpage
\appendix

\section{Wrapping Lemma}

For completeness, we add the version of the wrapping lemma used in
this paper. The
original lemma of Felsner~\etal \cite{FLW-GD01-ref, FLW-JGAA03} is the
below lemma with $\ell=1$. \figref{tree-wrap} below also depicts $\ell=1$ case. 

\begin{lem}\lemlabel{wrap}\cite{FLW-JGAA03, DPW-DMTCS04}
Let $T$ denote a track layout of a graph $G$ with tracks
in $T$ indexed by $(i,k)$ where $i\in\{0,\dots, p\}$ and $k\in
\{1,\dots, \ell\}$ and such that for each edge $vw$ of $G$, with $v$ in
track $(i_v, k_v)$ and $w$ in $(i_w, k_w)$, 
 $|i_v-i_w|\leq 1$. Then $T$ can be modified (wrapped) into a $3\ell$ track layout $T'$
of $G$ as follows:
Each vertex $v$ of $G$ is assigned to a
track $(i_v \mod 3, k_v)$ and two vertices $v$ and $x$ that are in the
same track of $T'$ are ordered as follows. Let $i_v\leq i_x$.
~\\ 
(1) If $i_v<i_x$, then $v <_{T'} x$. \\
(2) Otherwise, ($i_v=i_x$),  $v$ and $x$ are ordered in $T'$ as in $T$.
\end{lem}

\begin{proof}
It is simple to verify that each track in the track assignment $T'$ is
a total order. It remains to prove that $T'$ has no
X-crossings. Assume by contradiction that $vw$ and $xy$ form an
X-crossing in $T'$. Without loss of generality let $v<_{T'}x$ and
$y<_{T'} w$.  Since $i_v \mod 3=i_x \mod 3$, either $i_v=i_x$, or $i_x\geq
i_v+3$. 

If $i_v=i_x$, then $v<_{T} x$.  Since $v$ is adjacent to $w$,
$i_w=\{i_v-1, i_v, i_v+1\}$ and similarly $i_y=\{i_x-1,
i_x,i_x+1\}$. No pair of these indices differs by at least 3 when
$i_v=i_x$, thus $|i_w-i_y|\leq 2$. Since $w$ and $y$ are in the same track
in $T'$, $i_w \mod 3=i_y \mod 3$. Together with $|i_w-i_y|\leq 2$, this implies that  
$i_w=i_y$. That implies further that $y$ and $w$ are in the same track in $T$
and thus by  property $(2)$ their
ordering $y<_{T'}w$ in $T'$,  is the same as in $T$. Thus $y<_{T}w$, which implies that
$vw$ and $xy$ form an X-crossing in $T$ thus contradicting the
assumption. 

If $i_v<i_x$, then $i_x\geq i_v+3$. Again, since $v$ is adjacent to $w$, $i_w=\{i_v-1, i_v, i_v+1\}$
and similarly $i_y=\{i_x-1, i_x,i_x+1\}$. Thus $i_w\leq i_v+1$ and
$i_y\geq i_v+2$. Therefore, $i_w< i_y$ and by property $(1)$, $w<_{T'}
y$, contradicting the above assumption on the ordering of $y$ and $w$
in $T'$.
\end{proof}

\begin{figure}[h]
  \begin{center} 
   \includegraphics[width=0.9\linewidth]{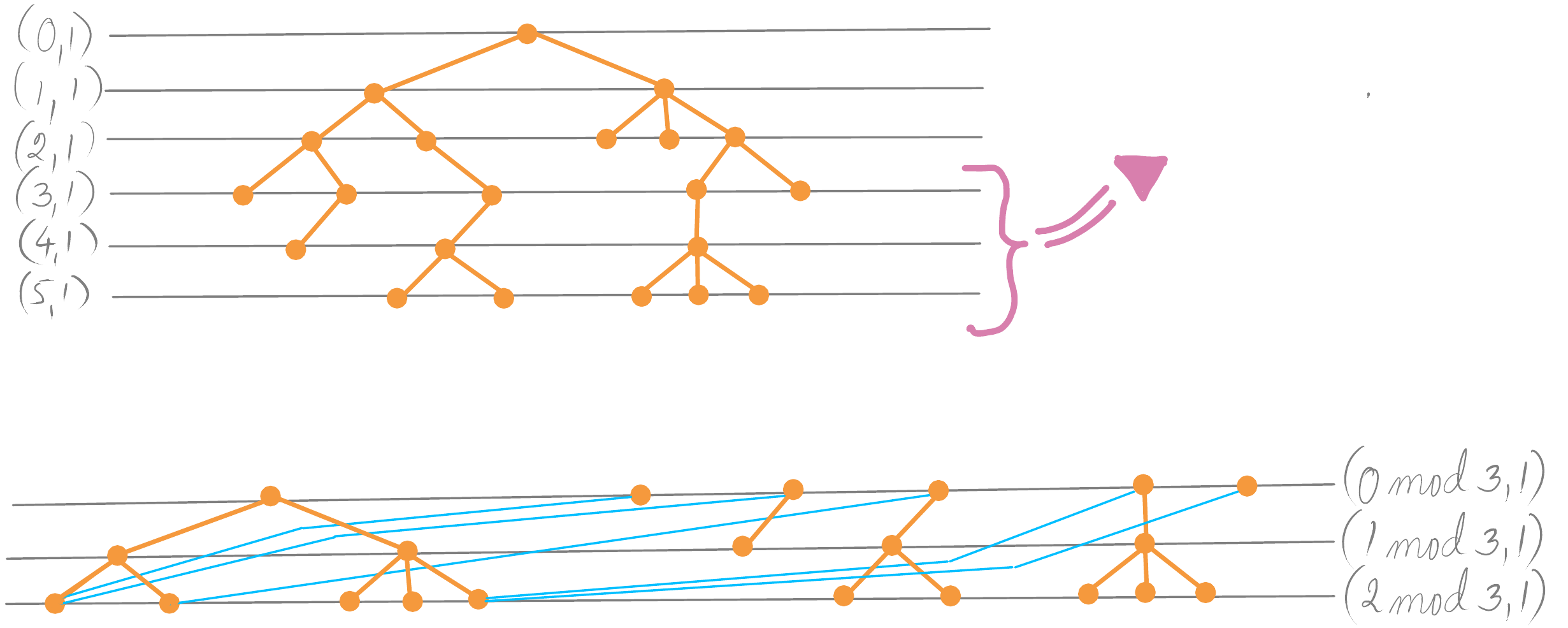}
  \end{center}
  \caption{Top figure: a $6$-track span-$1$ layout of a
    tree. Bottom figure: after wrapping to a $3$-track layout.}
  \figlabel{tree-wrap}
\end{figure}

\end{document}